\pdfoutput=1

\documentclass{comjnlarxiv}
\usepackage{booktabs}

\usepackage{float}
\usepackage{flushend}
\usepackage{url}
\usepackage{setspace}
\usepackage{listings}
\usepackage{multirow} 
\usepackage{ulem}
\usepackage{subfig}
\usepackage{amssymb}
\usepackage{amsmath}
\usepackage{algorithm}
\usepackage{algorithmic}
\usepackage{verbatim}
\usepackage{mdwlist}

\usepackage{enumitem}
\usepackage{pgf}
\usepackage{tikz}
\usetikzlibrary{arrows,automata}
\usepackage{tabularx}
\usepackage{amsthm}

\newcommand{\sch}{\mathop{\mathrm{sch}}}
\newcommand{\nullname}{\texttt{null}}

\newcommand{\literaljoin}{\widehat{\Join}}
\newcommand*{\card}[1]{\ensuremath \mid #1 \mid}
\newtheoremstyle{named}{}{}{\itshape}{}{\bfseries}{.}{.5em}{#1 #3}
\theoremstyle{named}
\newtheorem*{condition}{Condition}

\begin{document}

\title[Functional dependencies 
  with \nullname{} markers]{Functional dependencies 
  with \nullname{} markers}
\shortauthors{A. Badia and D. Lemire}

\author{Antonio
  Badia} \email{antonio.badia@louisville.edu} \affiliation{CECS
  Department, University of Louisville, Louisville KY 40292, USA}\address{} 
\author{Daniel Lemire} \affiliation{LICEF, Universit\'e du Qu\'ebec, 
 5800 Saint-Denis, Montreal, QC, H2S 3L5 Canada}


\begin{abstract}
Functional dependencies are an integral part of database
design. However, they are  only defined when we exclude \nullname{}
markers. Yet we commonly use \nullname{} markers in practice. 
To bridge this gap between theory and practice, researchers have
proposed  definitions of functional dependencies over relations  
with \nullname{} markers. Though sound, these definitions lack some
qualities that we find desirable.
For example, some fail to satisfy Armstrong's axioms---while these axioms are part of the
foundation of common database methodologies. 
We propose a set of properties that any extension of functional 
dependencies over relations with \nullname{} markers should possess. We then propose two new
extensions having these properties. 
These extensions attempt to allow \nullname{} markers where they 
make sense to practitioners.
 They both support Armstrong's axioms and
provide \emph{realizable}  \nullname{} markers: at any time, some or
all of the  \nullname{} markers can be replaced by actual values
without causing an anomaly.
Our proposals may improve database designs.
\end{abstract}

\maketitle
\section{Introduction}\label{sec:intro}   
Functional dependencies (hereafter, FDs) are the basis of good
relational database design.    Database designers use FDs  to define
and enforce consistency; for example, if each user account has only one
  corresponding primary email address, we say that the user account
  determines the email address ($\textrm{user-account} \to \textrm{email}$).
   Though database
  designers may not always work directly with FDs, almost all of them
  work with normal forms which exist to support FDs. 
     
Meanwhile,  relational databases use
\nullname{} markers to cope with incomplete information. Indeed, Codd introduced
\nullname{} markers in the relational model along with a 3-value
logic: a comparison between a value and a \nullname{} marker has an
unknown truth value~\cite{Codd:1986:MIR:16301.16303}. 
A \nullname{} marker can indicate anything from an applicable but
unknown value to a non-applicable attribute; it may  even indicate that we have 
\emph{no
  information}~\cite{Zaniolo1984142,Atzeni:1984:FDR:1147.1160,Hartmann:2012:IPD:2188349.2188355}.  We
 refer the reader to Libkin~\cite{libkin1994aspects} for an early
 survey on incomplete information within databases.

Unfortunately, the concept of functional dependence, even though central to
relational database design, is not part of the SQL standard.  Hence,
FDs are not \emph{directly} supported by relational
databases. Admittedly, given the procedural extensions of the SQL
standard, SQL is computationally complete and therefore can enforce
FDs  using checks, assertions or triggers. However, the problem is
more fundamental: FDs are not \emph{defined} in the presence of
\nullname{} markers. Thus, there is no clear
semantics to enforce when we have FDs and \nullname{} markers.

We believe  that the fact that FDs are not defined in
the presence of \nullname{} markers undermines the role of FDs in 
database design~\cite{Badia:2011:CAR:2070736.2070750}. While there are  proposed definitions of FDs in
relations with \nullname{} markers,  such proposals may not be  practical.

Our work is organized as follows. In the next section, we argue for a
minimal set of desirable properties that FDs in relations with
\nullname{} markers should have  to be of practical use. After briefly
introducing our formal notation in \S~\ref{sec:basic},  
 we review two existing extensions of FDs to \nullname{} markers in
 \S~\ref{sec:strongandweak}: weak and strong FDs. We show that neither
 extension is satisfactory in our context, as they do not have the
 desired properties.   
We introduce two  novel extensions in \S~\ref{sec:litandsr}: literal and
super-reflexive functional dependencies, and show that they that
possess our properties. In \S~\ref{sec:comparison} we compare our
 extensions to weak and strong FDs, in order to give the reader a
 context for interpretation. Next in \S~\ref{sec:complexity} we
 address the issue of whether the newly proposed concepts can be
 efficiently supported. In \S~\ref{sec:enforcing}, we show how to
 extend logical database design to include \nullname{}
 markers using one of the proposed extensions. Finally, we close
 in  \S~\ref{sec:conclusion} with some  comments and a discussion of
 further research.  
 
\section{Desirable properties}
\label{sec:motivation}

The SQL standard allows \nullname{} markers, but does not allude to
 FDs~\cite{sql-standard-2008}.  
 One could point out that the concept of \emph{key} is explicitly present in
the SQL standard; given that the concepts of key and FD are related,
 it would seem that we support FDs by enforcing 
the \emph{key constraint}~\cite{Hartmann01072011}.
 However, when we enforce  FDs by key constraints, we  assume that all
 the tables in the database have attained Boyce-Codd normal form (BCNF). 
Yet such a normal form excludes the commonly used \nullname{} markers.
 
Codd was well aware of the perceived problems  with \nullname{}
markers. Yet  he was unconcerned by our
inability to enforce FDs in the presence of \nullname{} markers. He
believed that only when the \nullname{} markers where replaced by
actual values would concepts such as keys, normalization and FDs be
applicable:
\begin{quote}
It should be clear that, because nulls (or, as they are now called,
marks) are NOT database values, the rules of functional dependence---and of multi-valued dependence---do not apply to
them. (E. F. Codd~\cite{Codd:1986:MIR:16301.16303}) 
\end{quote}
 However, consider the relation in Table~\ref{FirstTable}
subject to the FDs {$\text{professor} \to \text{chair}$} and 
  {$\text{department} \to \text{chair}$}. Given that Jill and Arthur are distinct
people, it is  not possible to replace the \nullname{} marker  by an
actual value.  Going back to Codd's vision, we have that keys,
normalization and FDs are never going to be applicable to all tuples of such relation. This might reasonably be
considered anomalous.  

 To avoid such problems, we establish as a first goal that if a
 FD is enforced in a relation with \nullname{} markers, it should
 always be  possible to replace some or all of the  \nullname{}
 markers by actual values without violating the FD ({\bf G1}).  In
 such a case,  we say that such  FDs are \emph{realizable} under null
 markers. Moreover, and since  we consider \emph{sets} of FDs
 for design (as opposed to single FDs in isolation), we posit that FDs
 should be defined so that if a \emph{set} of FDs is
 enforced in a relation with \nullname{} markers, it should be
 possible to replace some or all of the  \nullname{} 
 markers 
 without violating any of the
 FDs in the set ({\bf strong G1}).

\begin{table}
\centering
\caption{\label{FirstTable} Example relation with attributes
  professor, chair and department}   
\begin{tabular}{ccc}\toprule
professor & chair & department\\\midrule
Joe     & \nullname{}  & Mathematics\\
Joe     & Jill  & Computer Science\\
Bill    & Arthur& Mathematics \\\bottomrule
\end{tabular}
\end{table}

When defining
FDs in the presence of \nullname{} markers, we are also interested in
Armstrong's three axioms, especially transitivity:
\begin{enumerate*}
\item Reflexivity: If $Y\subseteq X$, then $X\to Y$. 
\item  Augmentation:  we have that  $X\to Y$ implies  $XZ\to
YZ$. 
\item Transitivity:  If $X \to Y$ and $Y \to Z$, then $X\to Z$. 
\end{enumerate*}
E.g., we might say that your social security number (SSN) determines
your income ($\text{SSN} \to \text{Income}$), and that your income determines your
tax bracket ($\text{Income} \to \text{Taxation}$). If transitivity holds, then your
SSN determines your tax bracket ($\text{SSN} \to \text{Taxation}$).

\begin{table}
\centering
\caption{\label{SecondTable} Example relation with attributes
  SSN, income and taxation}
\begin{tabular}{ccc}\toprule
SSN & income & taxation\\\midrule
1112233     & \nullname{}  & 15\%\\
1112233    & \nullname{}  & 25\%\\\bottomrule
\end{tabular}
\end{table}

Codd's interpretation, that FD do not apply when there are
\nullname{} markers  fails to enforce transitivity in the following
sense: given the  FD $\textrm{SSN} \to \textrm{Income}$ and
$\textrm{Income}\to \textrm{Taxation}$, both tuples in
Table~\ref{SecondTable} are allowable, even though one would expect
not to see such data in the database. It implies that
$\textrm{SSN} \to \textrm{Taxation}$ does not hold (whereas it should
under transitivity) even though there is no \nullname{} marker over
attributes $\textrm{SSN}$ and $\textrm{Taxation}$.  

Codd would no doubt reply that there is no violation of transitivity
since FDs do not apply in the presence \nullname{} markers. But we
wish to consider \nullname{} markers as an integral part of the
database. 

Failing to enforce Armstrong's axioms has significant 
consequences. For one thing, without these axioms, normalization is no
longer sufficient to enforce FDs. In practical terms, any redefinition
of the FDs that fails to satisfy  Armstrong's axioms cannot be
enforced through normalization.  Indeed, given a database $D$ and a
set of FDs  ${\cal F}$ on $D$, before we can use ${\cal F}$ to 
determine normal forms for the relations in $D$, we need to make sure
that ${\cal F}$ is in minimal or canonical
form~\cite{bernstein1976synthesizing}. 
But minimizing ${\cal F}$ depends \emph{crucially} on FDs
respecting transitivity, as covered in standard database textbooks.

We might be willing to forgo normalization and enforce FDs
through other means. In such a case, it might seem like Armstrong's
axioms are no longer required. For example, we might think that it is
possible to build a database design without assuming that FDs are
transitive. However, such watered-down FDs might be impractical for
other reasons.  The first problem that we encounter is that standard
database design methodologies, like the entity-relationship model,
implicitly assume Armstrong's axioms and transitivity in 
particular~\cite{Chen:1976:EMU:320434.320440}.

We believe that database designers would have a hard
time coping with the lack of transitivity (see, for instance, the
example of Table~\ref{SecondTable}, where intuitively one would expect
to see, for the same SSN, the same Taxation), and hence we require
that Armstrong's axioms hold, even though such a requirement could be
seen as too strong. Of course,  we could help designers 
with additional tools and methodologies~\cite{VanBaoTran2013} to
compensate for the added constraint. Nevertheless, everything else
being equal, we view as desirable that  a new definition of FDs in the
presence of \nullname{} markers should respect Armstrong's axioms
({\bf G2}), as well as enforce realizable \nullname{} markers.


Of course, our goals so far can be accomplished by being very restrictive on the
use of  \nullname{} markers. One might even take the stance that
\nullname{} markers should always be
forbidden~\cite{Date:2009:SRT:1550741}. But we 
also want to allow common uses seen in production-quality
applications. For example, we have observed that many database schemas
allow \nullname{} markers on attributes that do not determine other
attributes.   Thus, it is another objective~({\bf G3}) of this research
to allow \nullname{} markers when this can be done without violating
other goals. At a minimum, we should allow \nullname{} markers without
having to come up with contrived examples. 

Finally, any enforcement of FDs is going to be considered, from the 
point of view of transaction or query processing, as overhead---just
like enforcing primary and foreign key constraints. Thus, any
definition should  be computationally inexpensive ({\bf G4}). In
practice, this means that we exclude elegant but
challenging  models  such as
v-tables~\cite{Imielinski:1984:IIR:1634.1886}. For example,  we should be
able to determine whether the FD $X \to Y$ holds  by considering \emph{only}
the attributes in $X$ and $Y$. 

To summarize, we seek to extend FDs to include \nullname{} markers
in such a way that: 
\begin{enumerate*}
\item {\bf G1}: FDs enforce realizable \nullname{} markers. Further,
  this should hold for sets of FDs ({\bf strong G1}). 
\item {\bf G2}: Armstrong's axioms are satisfied. 
\item {\bf G3}: FDs should not restrict the use of \nullname{} markers
  unnecessarily.  
\item {\bf G4}: Enforcing FDs should be computationally  practical.
\end{enumerate*}

\section{Basic concepts}
\label{sec:basic}
Let a \emph{relation} $r$ be as  in SQL:
 a finite multiset of tuples over a given schema $\sch(r)$, with the
caveat that a tuple may contain \nullname{} markers. 
Two tuples are considered \emph{duplicates} if all non-\nullname{} attributes
are equal and any \nullname{} marker in one tuple is matched by
 a \nullname{} marker in the other tuple; otherwise the tuples
 are \emph{distinct}.   

We assume that there is an infinite set $V$ of values, from where all
the constants in any relation are drawn.
These values have a relation '=' defined on them, which is reflexive,
symmetric and transitive: given any $x,y, z\in V$, we have than $x =
x$, $(x=y )\Rightarrow (y = x)$ and $(x = y, y = z) \Rightarrow x =
z$. Hence, the relation ``='' is an equivalence relation.

Given an attribute $A$ in the schema of relation $r$ and a tuple $t$, we
use $t[A]$ as is customary, to denote the 
value of $t$ for $A$. This is extended to sets of attributes $X
\subseteq \sch(r)$ as usual. We then say, for two tuples $t$, $t'$, 
that $t[A] = t'[A]$ is \emph{true} if both $t[A]$ and
$t'[A]$ are equal non-\nullname{} values; \emph{false}
if both $t[A]$ and $t'[A]$ are values, but they are different;
and \emph{unknown} otherwise (i.e., if either one of $t[A]$ or
$t'[A]$, possibly both, are \nullname{} markers). Again, this is
extended to sets of attributes $X \subseteq \sch(r)$ in the usual way:
for two tuples $t$, $t'$,  $t[X] = t'[X]$ is \emph{true} if $t[A] =
t'[A]$ is true for every $A \in X$; \emph{false}  if $t[A] =
t'[A]$ is false for some $A \in X$,  and \emph{unknown} otherwise. As
a shorthand, we write $t=t'$ for $t[\sch(r)]=t'[\sch(r)]$.  We write
$\pi_X(r)$ for the projection of all tuples in $r$ on $X$: starting
from $\{t[X]|t\in r\}$, all duplicates are removed.

Let $r$ be a fixed relation. If we disallow \nullname{} markers in
$r$, then  a FD $X \to Y$ is satisfied if $t[Y]=t'[Y]$ when two tuples $t,t'$
are such that $t[X]=t'[X]$.
In this context (where \nullname{} markers are forbidden), FDs satisfy
Armstrong's axioms. We can also formalize the concept of key if there
is no \nullname{} marker in $r$. A set of attributes $K$ is a
\emph{superkey} 
 iff $K \to A$ holds for any attribute $A \in \sch(r)$;
and a  \emph{key} if it is 
a minimal superkey. \emph{Primary keys} in SQL are keys with attributes
where \nullname{} markers are forbidden;
SQL allows \nullname{} markers in other keys. 

We say that an
attribute $A$ is \emph{non-\nullname{}} in a relation $r$ if for all $t \in
r$ we have that $t[A]$ is non-\nullname{}. 
Given a set of attributes $X$, 
we say that the 
\nullname{} marker  appearing in a tuple $t$ at attribute $A$
($t[A]=\nullname{}$) is \emph{in $X$} if $A \in X$.

\section{Strong and weak functional dependencies}
\label{sec:strongandweak}
Levene and Loizou~\cite{levene1998axiomatisation} propose one of the
few extensions of FDs over \nullname{} markers.
We formalize their definitions as follows.
A \emph{valuation} $\varphi$ for a relation $r$ assigns to each 
\nullname{} marker in a tuple of $r$ a value from $V$---each \nullname{} marker
may receive a different value---while leaving non-\nullname{}
values from $V$ unchanged.
Given a relation $r$, each $\varphi(r)$ is called a \emph{possible world} for
$r$. The semantics of FDs with \nullname{} markers, as defined by
Levene and Loizou, 
follows the idea of modal logic~\cite{chellas1980modal}: we have two distinct readings,
depending on whether the FD holds in some or all possible worlds.

\begin{definition}
A FD $F$ holds weakly in relation $r$ iff $F$ holds in a possible world
for~$r$---i.e., there exists a valuation $\varphi$ such that $F$
holds in $\varphi(r)$. $F$ is called a \emph{weak FD} (WFD).
\end{definition}

\begin{definition}
A FD $F$ holds strongly in relation $r$ iff $F$ holds in all possible worlds
for~$r$---i.e., for every valuation $\varphi$ for $r$,  $F$
holds in $\varphi(r)$. $F$ is called a \emph{strong FD} (SFD).
\end{definition}

Looking back at  Table~\ref{FirstTable}, consider the following  FDs:
$\textrm{chair} \to \textrm{professor}$ and
$\textrm{professor} \to \textrm{chair}$. Both hold weakly, while
neither holds strongly (see Table~\ref{table:whichhold}).

\begin{table}\centering
\caption{\label{table:whichhold}Various FDs applied to the relation of
Table~\ref{FirstTable} and whether they hold  strongly, weakly,
super-reflexively or literally} 
\begin{tabularx}{0.9\columnwidth}{lXXXX}\hline
         & SFD & WFD & SRFD & LFD \\\hline
$\textrm{chair} \to \textrm{professor}$  & no &  yes & no & yes \\
$\textrm{professor} \to \textrm{chair}$  &  no & yes & yes & no \\
\hline
\end{tabularx}
\end{table}

A strong FD is always also a weak FD\@. Strong FDs
satisfy Armstrong's axioms while weak FDs do not. When a FD holds
strongly, we can replace any \nullname{} marker by a value, and
the FD  still holds. In fact, strong FDs enforce realizable
\nullname{} markers. 

The Levene and Loizou model has a substantial formal appeal,
but it also has some drawbacks from a pragmatic point of view: 
\begin{itemize}

\item Weak FDs might be too weak.
Even if each FD in a set $\mathcal{F}$ of FDs holds
weakly, there might be no single possible world $\varphi(r)$ where all
of the  FDs 
in the set $\mathcal{F}$ hold. See Table~\ref{FirstTable} for a
counterexample: both FDs ($\textrm{professor} \to \textrm{chair}$  and
$\textrm{chair} \to \textrm{professor}$) hold weakly, but there is no possible world
where they both hold.
That is, we can substitute some value for the \nullname{} marker   to satisfy
the first FD, and substitute \emph{another} value to satisfy the second FD,
but no single value  makes both FDs true at the same time (thus,
$\mathcal{F}$ does not have property {\bf strong G1}). 
This is true even when
each attribute appears only once on the right-hand-side of a FD\@:
given the schema $A,B,C$ and the FDs $A\to B$ and $B\to C$, the set of
tuples  $(a,\nullname{},b)$ and $(a,\nullname{},c)$ satisfy both FD weakly, but
there is no world where they both hold. 
 This last example also illustrates that weak FDs are not
transitive: $A \to B$ and $B \to C$ can hold weakly while $A \to C$
may not. That is, weak FDs do not satisfy
  Armstrong's axioms (thus failing {\bf G2}). 

We could \emph{fix} weak FDs to ensure that they enforce, for example,
realizable \nullname{} markers by requiring that there exists a
valuation corresponding to the  set of all FDs. However, this may
prove computationally challenging (hence failing {\bf G4}). 

  \item Strong FDs might be too strong (thus failing {\bf
    G3}). Indeed, it seems unnecessary to always require that FDs
    should hold in all possible worlds. For example, consider 
  the schema $A,B,C$ and the FDs $A\to B$ and $B\to C$, and the set of
tuples  $(a,b, \nullname{})$ and $(c,b, \nullname{})$. Though it
appears like a reasonable relation, the FD $B\to C$ fails to hold
strongly.  
\end{itemize}

In some sense, the weak and strong FDs are two extremes, whereas the
right solution might be somewhere in between.  Indeed, any form of FD
that supports realizable \nullname{} markers ({\bf G1}) on a per FD
basis, is going to be equivalent to or stronger than weak
FDs. Meanwhile, strong FDs have the properties we seek, except that
they are too restrictive ({\bf  G3}). 

\section{Literal and Super-reflexive functional dependencies}
\label{sec:litandsr}

We propose two alternative definitions of the concept of FD in the
presence of \nullname{} markers. The first one is a natural extension
of the 3-value logic proposed by Codd and used by SQL. 

\begin{definition}\label{def:SRFD}
A FD  $X \to Y$ holds  super-reflexively if, for any two tuples $t,t'$, 
 when 
$t[X]=t'[X]$ is not false (i.e., it is either true or unknown), then
$t[Y]=t'[Y]$ is also not false. We
say $X \to Y$ is a \emph{super-reflexive} FD (SRFD). 
\end{definition}
In this first definition, the \nullname{} marker is effectively  
equal to any other value (i.e., $\nullname{} = a$ is treated as
true for any value $a$), hence the term \emph{super-reflexive} (SR).

As an illustration, consider Table~\ref{FirstTable}. We have that
$\textrm{professor} \to \textrm{chair}$  hold super-reflexively whereas 
$\textrm{chair} \to \textrm{professor}$ does not (see Table~\ref{table:whichhold}).

Our second definition is reminiscent of how languages such as
JavaScript handle \nullname{} markers. As a first approximation, they
consider \nullname{} to be effectively a regular value, with
$\nullname{} = \nullname{}$ is always true (i.e., ``='' remains a reflexive
relation\footnote{In fact, '=' remains an equivalence relation,
  something that does not hold for Codd's 3-value logic if you
  consider  \nullname{} markers to be part of the value domain.}), but 
$\nullname{} =a$ always false for $a$ non-\nullname{}.   

We say that $t[A]$ and $t'[A]$ are \emph{identical} if both contain
\nullname{} or both contain the same value; this is also
  extended to set of attributes $X$ and to whole tuples as usual:
  $t[X]$ is identical to $t'[X]$  if and only if $t[A]$ and $t'[A]$
  are identical for all $A\in X$.

\begin{definition}\label{def:LFD}
A FD  $X \to Y$ holds  literally if, for any two tuples $t,t'$, when  
$t[X]$ and $t'[X]$ are  identical then $t[Y]$ and $t'[Y]$ are also
identical. We say $X \to Y$ is a \emph{literal} FD (LFD). 
\end{definition}

Consider again Table~\ref{FirstTable}. In contrast with the
super-reflexive case, we have that the FDs $\textrm{chair} \to
\textrm{professor}$  holds literally whereas  $\textrm{professor} \to
\textrm{chair}$ does not. (See again Table~\ref{table:whichhold}.)

There are alternative definitions that we could have used. For
example, we could have defined FDs $X \to Y$ to hold if   when 
$t[X]=t'[X]$ is true (as per Codd's 3-value logic) then $t[Y]=t'[Y]$
must be true. Or, we could have defined an FD $X \to Y$ to hold
if when $t[X]=t'[X]$ is not false then  $t[Y]=t'[Y]$ must be true; or 
if when $t[X]=t'[X]$ is true then $t[Y]=t'[Y]$ must be not false.
However, these  alternative definitions are unsatisfying: they fail to
satisfy Armstrong's axioms ({\bf G2}), or do not allow  \nullname{}
markers where they are commonly used ({\bf G4}).
By contrast, Definitions~\ref{def:SRFD} and~\ref{def:LFD} have the
properties that we 
required in \S~\ref{sec:motivation} starting with Armstrong's axioms
(which follows by inspection). For example, regarding transitivity,
consider the two FDs $X\to Y$ and $Y \to Z$ that hold
super-reflexively (resp.\ literally). Given two tuples $t,t'$ such that
$t[X]=t'[X]$ is not false (resp. $t[X]$ and $t'[X]$ are identical),
then  $t[Y]=t'[Y]$  is not false (resp. $t[Y]$ and $t'[Y]$ are
identical) which implies that $t[Z]=t'[Z]$ is not false (resp. $t[Z]$
and $t'[Z]$ are identical). Thus we have that $X \to Z$, proving
transitivity. 

\begin{lemma}
Super-reflexive and literal FDs respect Armstrong's axioms ({\bf G2}). 
\end{lemma} 

Before we proceed to establish other properties, we need a 
technical result that makes other proofs easier. First, 
we can verify that all FDs can be decomposed into FDs where the right-hand-side
contains a single attribute.   For example, the FD $\textrm{professor}
\to \{\textrm{chair}, \textrm{department}\}$ is equivalent to the
following two FDs: 
\begin{itemize} 
\item $\textrm{professor} \to  \textrm{chair}$, and 
\item $\textrm{professor} \to \textrm{department}$.
\end{itemize}

\begin{lemma}
We have that  $X \to Y$ for $Y=\{A_1,\ldots A_n\}$ is
 equivalent to $(X-\{A_1\}  \to \{A_1\}) \land \cdots \land (X-\{A_n\}
 \to \{A_n\})$ whether we consider 
 weak, strong, literal or super-reflexive FDs. \end{lemma}

The proof of this lemma follows by inspection. Hence, it is enough
 to consider FDs $X \to Y$ where  $Y$ is a singleton disjoint from $X$. 

First, we must show that SRFDs and LFDs satisfy condition {\bf
  G1}: \nullname{}
 markers are realizable.  We begin with SRFDs. We show
that given the SRFD $X\to Y$, all \nullname{} markers in $X \cup Y$
are realizable with respect to $X \to Y$. 
   To illustrate this result, consider Table~\ref{FirstTable}
where $\textrm{professor} \to \textrm{chair}$ holds super-reflexively:
we can substitute  $\textrm{Jill}$ for the  \nullname{} marker without
violating the FD.  

\begin{lemma}\label{lemma:techlemma}
Consider a SRFD $X\to Y$ in a relation  such that $X$ and $Y$ are
disjoint and $Y$ is a singleton. 
\begin{enumerate}
\item \label{leftsafe} We can replace any \nullname{} marker in $X$
  by   \textbf{any} actual value without violating the SRFD $X\to Y$. 
\item \label{rightsafe} We can replace any \nullname{} marker in $r$ in $
Y$ by \textbf{at least one} actual value without violating the SRFD
$X\to Y$, assuming that attributes in $X$ are non-\nullname{}.
\end{enumerate}

\end{lemma}
\begin{proof}
Assume that the SRFD is initially satisfied  over $r$. 

(\ref{leftsafe}) Suppose we replace a \nullname{}~marker in
  attribute $B\in X$. Regarding the SRFD $X\to Y$, the following might happen  for two tuples
  $t, t'$:
  \begin{itemize}
  \item  if  $t[X]=t'[X]$ was false, then it is still false: this cannot affect the SRFD;
  \item if $t[X]=t'[X]$ was true, then it is still true: this cannot affect the SRFD;
  \item if $t[X]=t'[X]$ was unknown, then it might become true or false. Because of the SRFD $X\to Y$, and because  $t[X]=t'[X]$ was not false, we have that $t[Y]=t'[Y]$ is not false. Therefore, if the tuples $t,t'$ satisfied the SRFD before the update, they must satisfy it after the update as well.
  \end{itemize}
Because all pairs of tuples satisfy the conditions of SRFD after the update, the SRFD still hold, proving the first part of the result.
%
%

(\ref{rightsafe})  
Assume that attributes in $X$ are
non-\nullname{}. We want to show that we can replace any 
\nullname{} marker in $Y$ by an actual  value.  
  Pick a tuple $t$ in $r$ where $t[Y]$ contains 
a \nullname{}  marker.  Consider the set $\tau$ of $t'$ such that
 $t'[X]=t[X]$ is not false. (Because attributes in $X$ are
non-\nullname{}, we have that ``$t'[X]=t[X]$ is not false'' is
equivalent to  ``$t'[X]=t[X]$ is true''.)  We have
 that the projection of $\tau$ over $Y$ contains at most one actual
 value. (Suppose it does not, then you can find tuples $t''$ and
 $t'''$ in $\tau$ such that $t''[X]=t'''[X]$ is not false but
 $t''[Y]=t'''[Y]$ is false.) If there is one actual value, set $t[Y]$
 to this value; if not, pick a  value at random. 
This modification clearly does not violate the SRFD $X\to Y$ but it
eliminates one \nullname{} marker.
\end{proof}

To see why Lemma~\ref{lemma:techlemma} implies that SRFDs satisfy
{\bf G1}, consider any SRFD $X\to Y$ over a given relation. We can
substitute actual values for any \nullname{} marker in an attribute of
$X$ by the first part of the lemma. As a second step, since attributes
in $X$ have become non-\nullname{}, we can substitute actual values
for any \nullname{} marker in $Y$.

As for LFDs, we are going to prove something stronger: that they support {\bf strong G1}.

\begin{lemma}\label{lemma:ofsets1}LFDs strongly enforces realizable
  \nullname{} markers ({\bf strong G1}). 
\end{lemma}
\begin{proof} To prove {\bf G1}, it suffices to  replace all the
  \nullname{} markers with a single $v \in V$ not already in the
  relation. By inspection,  this extends to sets of FDs, so we
  get also {\bf strong G1} with this method.
\end{proof}

We have that both LFDs and SRFDs enforce realizable
\nullname{} markers. That is, given a relation with a set of FDs, we
can always replace \nullname{} markers with some actual values
without violating the FDs. In fact, if we add an extra constraint on
SRFDs, they both \emph{strongly} enforce 
realizable \nullname{} markers (in the sense of {\bf strong G1}). 
 In this context, we adopt the practical convention that
some attributes are allowed to contain \nullname{} markers while
others may not: this is motivated by the SQL standard. 
Given a set of FDs $\mathcal{F}$, we say that an attribute $B$
\emph{determines} another  attribute $A$ under $\mathcal{F}$ if there is a
FD $B\in X\to Y \ni A$  in the transitive closure of $\mathcal{F}$. Naturally, this property is transitive: if $A$ determines $B$ and $B$ determines $C$ then $A$ determines $C$.
(By convention, we omit loops in $\mathcal{F}$: $X \to X$.)

\begin{condition}[1RHS] Consider a set of FDs $\mathcal{F}$ over a
  relation. Consider any attribute $A$ allowed to
  contain \nullname{} markers. Then $A$ must appear on the
  right-hand-side of at most one FD in the set of FDs $\mathcal{F}$ of the relation. Moreover, given two
  distinct  attributes allowed to contain \nullname{} marker, $A$ and $B$, if $A$ determines $B$, $B$ cannot determine $A$.
\end{condition}

We stress that Condition~1RHS only applies to attributes allowed to contain \nullname{} markers: no constraint is required on other attributes. 

Fig.~\ref{fig:illustration} gives an example of a set of FDs satisfying the condition 1RHS even if all attributes are allowed to contain \nullname{} markers: $\{E,D\}\to \{A\}, \{A\}\to \{F\}, \{A,B\}\to \{F\}$. However, if we replaced the single FD $\{E,D\}\to \{A\}$ by two FDs such as $\{E\}\to \{A\}$ and $\{D\}\to \{A\}$, we would need to add the requirement that $A$ is non-\nullname{} to satisfy 1RHS\@. Similarly, if we added the FD 
$\{F\}\to \{A\}$ in addition to the existing FD $\{A\}\to \{F\}$, we would need to require that both $A$ and $F$ are non-\nullname{} since $F$ would determine $A$ while $A$ determines $F$.

\begin{lemma}\label{lemma:ofsets2}SRFDs strongly enforces realizable
  \nullname{} markers whenever the 1RHS condition is satisfied ({\bf
    strong G1}). 
\end{lemma}
\begin{proof}
Assume without loss of generality that all SRFDs in the set are of the
form $X \to Y$ where $Y$ is a singleton and $X,Y$  are disjoint.

\begin{figure}
\centering
\begin{tikzpicture}[thick,scale=0.7, every node/.style={transform shape},->,>=stealth',shorten >=1pt,auto,node distance=3cm,
                    semithick,main node/.style={circle,draw,font=\sffamily\Large\bfseries}]


  \node[state] (A)                    {$A$};
  \node[state]         (B) [above right of=A] {$B$};
  \node[state]       (C) [below right of=B] {$C$};
  \node[state]         (D) [below left of=C] {$D$};
  \node[state]         (E) [below left of=A]       {$E$};
  \node[state]         (F) [above left of=A]       {$F$};

  \path (A) edge            node {\small $\{A\}\to \{F\}$} (B)
        (B) edge           node {\small$\{A,B\}\to \{F\}$} (C)
        (D) edge              node {\small $\{E,D\}\to \{A\}$} (A)
        (E) edge   node {\small $\{E,D\}\to \{A\}$} (A)
        (A) edge   node {\small $\{A\}\to \{F\}$} (F);

\end{tikzpicture}
\caption{\label{fig:illustration}Illustration used by the proof of Lemma~\ref{lemma:ofsets2} for
  the set of FDs $\{\{E,D\}\to \{A\}, \{A\}\to \{F\}, \{A,B\}\to \{F\}\}$. }
\end{figure}
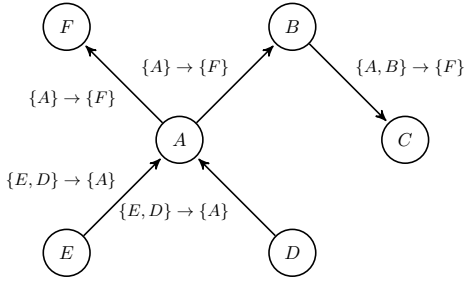
  Construct a graph where each attribute allowed to contain
  \nullname{} markers  is a node, and there is an edge between two
  attributes $A,B$ if and only if $A$ determines $B$.  
  We illustrate such a graph in Fig.~\ref{fig:illustration}  where, for simplicity, we omitted some of the edges that are implied by transitivity. 
  
  Temporarily
  assume that the graph is not empty. Because of condition 1RHS, the graph must be cycle-free and, therefore, some of the nodes must have a zero in-degree (e.g., $E$ and $D$ in Fig.~\ref{fig:illustration}). Call
  this set of nodes/attributes $\mathcal{A}_0$. 
  
  As per
  Lemma~\ref{lemma:techlemma}, we can substitute actual values for any
  \nullname{} marker they contain.  Indeed, consider such an  attribute $A\in \mathcal{A}_0$. This
attribute appears on the right-hand-side of at most   one FD in the
set (call it $F_A$), however $A$ may appear on the left-hand-side 
  of several FDs. These FDs are not a concern: replacing a \nullname{}
  marker in attribute $A$ may never violate a super-reflexive FD as
  per the first part of Lemma~\ref{lemma:techlemma}. Meanwhile,
  because $F_A$ is such that no attribute on its left-hand-side
  contains a \nullname{} marker, then the second part of
  Lemma~\ref{lemma:techlemma} tells us that the \nullname{} markers of
  $A$ are realizable. 

 After substituting actual values for any \nullname{} marker in the
 attributes of $\mathcal{A}_0$, remove these nodes from the
 graph. There must again be nodes with zero in-degree (e.g., node $A$ in Fig.~\ref{fig:illustration}) or the graph is
 empty. Repeat the process until the graph is empty. 
    
  Attributes that either do not appear as part of any FD, or that are not allowed to contain \nullname{} markers,  are not a
  concern.  \end{proof} 

\section{Comparing functional dependencies}
\label{sec:comparison}
We are now in a position to characterize LFDs and SRFDs properly. 
We can relate  LFDs and SRFDs to each other, and 
with Levene and Loizou's definitions. All are
\emph{conservative} extensions of the classical concept: in a table
without any \nullname{}, they coincide with classical FDs. However, in
general, LFDs and SRFDs are \emph{incomparable} as illustrated by Table~\ref{table:whichhold}.

We  show that 
 $\mbox{strong}
\Rightarrow   \mbox{super-reflexive}  \Rightarrow \mbox{weak}$
and  $\mbox{strong}
\Rightarrow   \mbox{literal}  \Rightarrow \mbox{weak}$ (see
Fig.~\ref{Fig:diag}). 
That is, both LFDs and SRFDs are stronger than weak FDs, whereas
strong FDs are stronger than  both LFDs and SRFDs.  

\begin{figure}\centering
\includegraphics[width=0.3\textwidth]{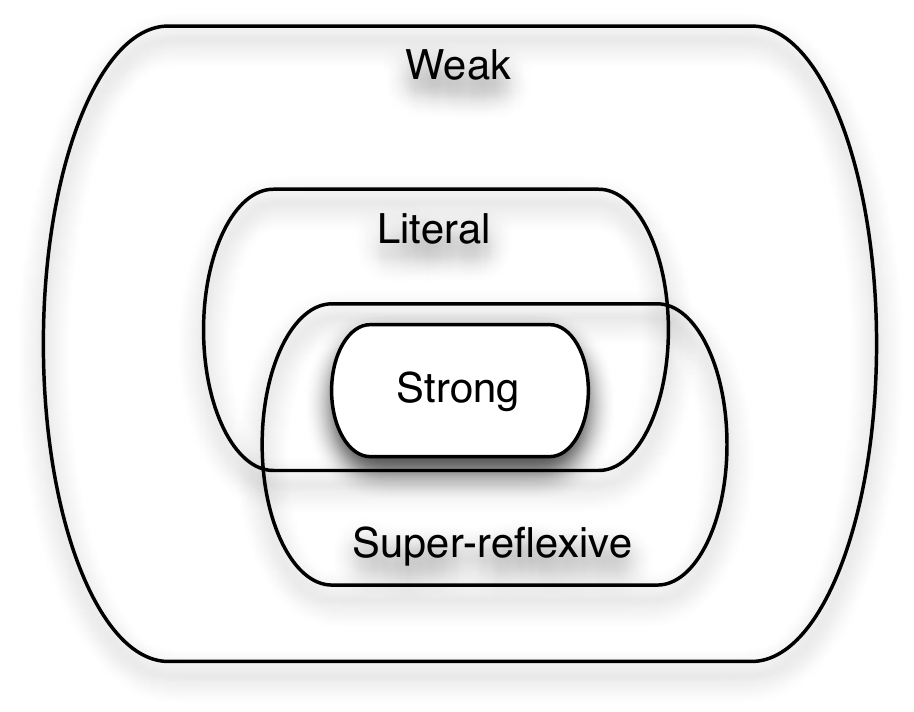}
\caption{\label{Fig:diag}Venn diagram illustrating
  Lemma~\ref{lemma:weaklemma}.} 
\end{figure}

 \begin{lemma}\label{lemma:weaklemma}
The following holds:
\begin{enumerate}
\item
\label{lit-weak}
If the FD $X \to Y$ holds literally,
then it holds weakly. 
\item
\label{sr-weak}
If the FD $X \to Y$ holds super-reflexively,
then  it holds weakly. 
\item
\label{strong-both}
If the FD $X \to Y$ holds strongly then it must hold 
 literally and super-reflexively.  
\end{enumerate}
 \end{lemma}

 \begin{proof}
Assume without loss of generality that $X$ and $Y$ are disjoint and
that $Y$ is a singleton. 

(\ref{lit-weak}) If the FD  $X \to Y$ holds literally then we can
   replace any \nullname{} marker by any one $v\in V$ not already
   present in the relation, and the FD still holds by inspection. This
   shows that literal FDs are stronger than weak FDs.  

(\ref{sr-weak})
By Lemma~\ref{lemma:techlemma}, we can construct a valuation such that the super-reflexive FD is valid in the
conventional sense. Because of the existence of the valuation, we have
 that $X \to Y$ holds weakly.   
 
(\ref{strong-both})
We first prove that strong implies super-reflexive.
Suppose that $X \to Y$ holds strongly. Consider two tuples $t,t'$. If
$t[X]=t'[X]$ is not false (in Codd's 3-value
logic), then in some possible world, $t[X]=t'[X]$
must hold.  In such a possible world $t[Y]=t'[Y]$ must hold which
implies that $t[Y]=t'[Y]$ must be not false (in Codd's 3-value
logic). This prove that a strong FD implies a super-reflexive FD. 

We prove that strong implies literal.
Suppose that the FD $X \to Y$  holds strongly and that
$t[X]$ and $t'[X]$ are identical. (We assume that $X$ and $Y$ are
disjoint and $Y$ is a singleton as previously stated.) There are some 
worlds where $t[X]=t'[X]$ is true. If $t[Y]$ and $t'[Y]$ are not identical,
then in at least one of these worlds, they would differ. That is, if
one has a \nullname{} marker in $t[Y]$ and a value $a$ in $t'[Y]$
(or vice versa) we can replace the \nullname{} marker by a value
that differs from $a$. Hence, we must have that $t[Y]$ and $t'[Y]$ are
identical given that $t[X]$ and $t'[X]$ are identical. Therefore
strong FDs imply literal FDs. This completes the proof. 
 \end{proof}

We can verify that LFDs and SRFDs are strictly weaker than SFD\@. Recall
that given the schema $A,B,C$ and the FDs $A\to B$ and $B\to C$,  the
set of tuples  $(a,b, \nullname{})$ and $(c,b, \nullname{})$ violates
the strong FD $B\to C$. However, it satisfies the FD $A\to B$ and
$B\to C$ literally and super-reflexively. The fact that LFDs and  
SRFDs allow this use of \nullname{} markers support our claim that
they do not unnecessarily forbid \nullname{} markers where they might
make sense ({\bf G3}). 

\section{Computational efficiency}
  \label{sec:complexity}
The remaining question is the efficient implementation of our
proposed concepts ({\bf G4}). Here we show that both LFDs and SRFDs
can be enforced with a cost similar to enforcing \emph{regular} FDs. 

Since SQL does not enforce FDs directly, we compare with a closely
associated property, the cost of enforcing the key constraint in a
relation. This assumes that the database is in an appropriate normal
form such that enforcing keys is equivalent to enforcing FDs. 
Given relation $r$, if $X \subseteq r$ is declared as the primary key for
$r$, any insertion $t \in r$ is checked to make sure that $X$ remains
a key. In practice, relational database systems  forbid the
occurrence of two distinct tuples $t,t'$ agreeing on $X$ ($t[X]=t'[X]$
must be false). This can be achieved by creating an index
(traditionally, a tree-based index) on $X$: on inserting $t$, we
search for any $t' \in r$ with $t[X] = t'[X]$ (note that \nullname{}
markers are forbidden in a primary key). Call the set of tuples
resulting from the search $S$. When inserting, if $S = \emptyset$, the
insertion can proceed; else, for each $t' \in S$ we check whether $t'$
is identical to $t$. If this is so, the insertion can proceed; else it
is forbidden. Updates are handled in a similar manner.  
The complexity of this procedure is considered
$\mathcal{O}(\log\card{r})$ when using a tree-based index, since it is
expected that $\card{S} \leq 1$. (Of course, we can get an expected
constant time complexity by  using a hash-based index.)  

We can also enforce FDs directly, whether they are  literal or
super-reflexive.  Given relation $r$ and arbitrary FD $X \to Y$ on it,
we create an index on $X$---for concreteness, we assume a $B^*$-tree
index since these are available on almost any database system. 

To enforce FDs with an index, it suffices to check whether the
insertion of a new tuple $t$ is allowed. Indeed, an update can be
viewed as a deletion followed by an insertion, and deletions cannot
violate a FD.

 Given an index,
enforcing LFDs is not difficult. We build a single index on $X$ by
considering a concatenation of all attributes---lexicographic order is
followed, with \nullname{} markers in individual attributes kept
together at the beginning or end of their positions. That is, in an
index for attributes $ABC$,   tuple $(a_1,b_1,\nullname{})$ would go
before or after all tuples   $(a_1,b_1,c)$, for $c$ any value of $C$;
tuple $(a_1,\nullname{},b_1)$   would go before or after any tuples
$(a_1,c,b_1)$ for $c$ any value   of $C$; and so on. This can be
achieved by considering  \nullname{} markers as either strictly larger
than any other value, or strictly smaller.  Suppose that we insert a
new tuple $t$. We can find in time $\mathcal{O}(\card{S}+\log
\card{r})$ the set $S$ of  all tuples $t'$ such that $t[X]$ and
$t'[X]$ are identical:  
  
\begin{align*}S = \bigcap_{A \in X  }     \{
  t'\in r | t'[A] = t[A] \enskip \lor \enskip \text{$t'[A]$  and
    $t[A]$ are \nullname{}} 
  \}\end{align*} 
   We can then check whether $t[Y]$ and $t'[Y]$ are identical for all
   $t'$ in linear time $\mathcal{O}(\card{S})$. Note that we consider the
   cardinality of the set $X$ ($\card{X}$) to be a small constant. 
  
We can also enforce $X \to Y$ as an SRFD efficiently, though the total
cost is probably larger. Index each one of the attribute $A \in X$
using, as before, the convention that y considering  \nullname{}
markers as either strictly larger than any other value, or strictly
smaller. Given  $t$,  first we check whether $t[Y]$ contains a
\nullname{} marker (assume $Y$ is a singleton disjoint from $X$), in
which case no work is needed. Otherwise, we find all $t'$ such that
``$t[X] = t'[X]$ is not false'' by computing 
\begin{align*}S=  \bigcap_{A \in X | t[A]\text{~not \nullname{}}}
  \{t'\in r | t'[A] = t[A]  \lor  t'[A] \text{~is
    \nullname{}}\}\end{align*}  
with the convention that the result is $r$ if $t[X]$ only contains
\nullname{} markers. 
Computing the intersection $S$ between several sets $S_1, S_2, \ldots,
S_{\card{X}}$ requires only complexity $\mathcal{O}(
\sum_{i=1}^{\card{X}} |S_i| + \card{S } )$ or
better~\cite{Ding:2011:FSI:1938545.1938550}. Each set $S_i$ can be
generated in time $\mathcal{O}(\card{S_i} + \log \card{r})$ for 
an overall complexity of $\mathcal{O}( \sum_{i=1}^{\card{X}} |S_i| +
\card{S } + \log \card{r} )$.  We then consider $t'[Y]$ for all $t' \in
S$: if there is an actual value, check that $t[Y]$ is equal to
$t'[Y]$. That is, we return
\begin{align*}\bigwedge_{t' \in S} t'[Y]\ = t[Y] \enskip \lor \enskip
  t'[Y]\text{~is \nullname{}}.\end{align*} 
  This can be computed in time $\mathcal{O}(\card{S})$.

To sum up, enforcing a LFD or SRFD  under an update or insertion can
be done in time  
\begin{itemize}
\item   $\mathcal{O}(\card{S}+\log \card{r})$ or
\item $\mathcal{O}( \sum_{i=1}^{\card{X}} |S_i| + \card{S } + \log
  \card{r} )$. 
\end{itemize}
Our sketched implementations only require tree look-ups and  set
intersections, both of which are well supported by all database
systems. This is sufficient to conclude that they are computationally
practical ({\bf G4}).

\section{Extending logical database design to include \nullname{} markers}
\label{sec:enforcing}

As we discussed in \S~\ref{sec:litandsr}, it is possible to enforce
FDs with \nullname{} markers using either SRFDs or LFDs---without any
particular effort on the part of the database designer. However, we
commonly enforce FDs using logical database design. That is, we
decompose relations into normal forms and identify  keys.  

We would like to remain as close as possible to the spirit of
SQL\@. Thus, we ask whether  we can use logical design  with SRFDs and
LFDs. This would allow an extension of logical database design to
include \nullname{} markers.  Unfortunately, while both LFDs and SRFDs
fulfill all our desiderata ({\bf G1} to {\bf   G4}), SRFDs are not
compatible with conventional logical design. But we have better
luck with LFDs.

Recall that in a relation $r_1$, a set of attributes $X \subseteq
\sch(r_1)$ is a \emph{key} if the (standard) FD $X \to \sch(r_1)$ holds
and if $X$ is minimal. Moreover, in relation $r_2$, a set of
attributes $Y \subseteq \sch(r_2)$ is 
a \emph{foreign key} for $r_1$ if $\pi_{Y}(r_2) \subseteq \pi_X (r_1)$ for all
extensions of $r_1$ and $r_2$.  

A join $r = r_1 \Join_{X=Y} r_2$ is a new relation made by combining
all tuples $t_1 \in r_1$ and all tuples $t_2 \in r_2$ such that
$t_1[X]=t_2[Y]$ into a new tuple $t$ equal to $t_1$ on $\sch(r_1)$ and
equal to $t_2$ on $\sch(r_2)-Y$. A join  $r = r_1 \Join_{X=Y} r_2$ is 
lossless when $\pi_{\sch(r_1)}(r) = r_1$ and $\pi_{\sch(r_2)}(r) = r_2$.

We extend the concepts of keys and joins in the context of LFDs as follows.

\begin{itemize}\item
 We say that  $X
\subseteq \sch(r)$ is a \emph{literal
  superkey} for $r$ if $X \to Y$ holds literally for any $Y \subseteq
\sch(r)$ and a \emph{literal key} iff it is  a  minimal literal superkey. 
 A \emph{literal foreign key} is
 an integrity constraint between two relations: a set of attributes
 $X$ in relation $r_1$ must match a set of attributes $X$ in a
 relation $r_2$ such that for every tuple $t$ in $r_1$, there must be
 a tuple $t'$ in $r_2$ such that $t[X]$
 and  $t'[X]$ are identical and such that $X$ contains a literal key
 in $r_2$.     
 \item
 As in SQL, each literal  foreign key constraint supports a
corresponding join. 
The literal join of $r_1$ and $r_2$ on $X$, a
foreign key in $r_1$, noted $\literaljoin{}$
is defined as follows: given any two tuples $t_1, t_2$ from $r_1, r_2$
such that $t_1[X]$ and $t_2[X]$ are
identical, we generate the tuple $t$ such that $ t[A] = t_1[A] $ for
all $A \in \sch(r_1)$ and $t[A] = t_2[A]$ for all $A \in
\sch(r_2)-\sch(r1)$. 
\end{itemize}

With these definitions, we have that lossless joins are supported,
even with \nullname{} markers. Formally, given relation $r$  with $\sch(r)= Z \cup W$ and such that $Z \cap W \to W$ holds literally,  $\literaljoin{}$ is a \emph{lossless join}:
$r= \pi_Z(r)  \literaljoin{} \pi_W(r)$.

\begin{proposition}(Lossless join) If $\sch(r)= Z \cup W$ and  $Z \cap W \to W$ holds literally then 
$r= \pi_Z(r)  \literaljoin{} \pi_W(r)$.
\end{proposition}
\begin{proof}
For the purpose of the literal join, the \nullname{}~marker can 
be treated like any other value. That is, the set of values $V$ extended with the \nullname{}~marker is effectively reflexive, symmetric and transitive. To conclude the proof, we  
have to show that $r=\pi_Z(r)  \literaljoin{} \pi_W(r)$.
\begin{enumerate}
\item Suppose that $t \in r$. There will be a tuple 
 $t^{(Z)}$ in $\pi_Z(r)$ such that $t[Z]$ and 
$t^{(Z)}$ are identical. Similarly, there will be a tuple
 $t^{(W)}$ in $\pi_W(r)$ such that $t[W]$ and 
$t^{(W)}$ are identical. We have that $t^{(W)}[Z \cap W]$ is identical to  $t^{(Z)}[Z \cap W]$. Thus we have that $t \in \pi_Z(r)  \literaljoin{} \pi_W(r)$.
\item Suppose that $t \in \pi_Z(r)  \literaljoin{} \pi_W(r)$.
There must be $t' \in r$ such that $t'[Z]=t[Z]$. We have that $t[Z \cap W]$ and
$t'[Z \cap W]$ must be identical because $Z \cap W \subset Z$. Because $Z \cap W \to W$, we have that $t[W]$ and $t'[W]$ are identical. Thus we have that $t[Z \cup W]$ and $t'[Z \cup W]$ which shows that $t \in r$.

\end{enumerate}
This concludes the proof.
\end{proof}

 In fact, we can show that logical
design is sound over LFDs: as long as we can normalize a relation in
the conventional sense, then we can normalize it in the sense of
LFDs. 

\section{Conclusion and Further Research}
\label{sec:conclusion}

We have reviewed the concept of FD in the presence of \nullname{}
markers, and have specified a set of properties that we believe any
definition of the concept should satisfy. We have proposed two new definitions of what it means for an
FD to hold in this situation for which the properties do hold: 
our definitions satisfy Armstrong's axioms for relations with \nullname{}
markers, allow \nullname{} markers to be used in practice (not only with contrived examples), and at the same time allows those
\nullname{} markers to be updated to real values consistently. These
FDs can also be enforced efficiently in computational terms  
despite \nullname{} markers. Both definitions have slightly different
properties (LFDs enforce lossless join, in addition to database
consistency), but they both satisfy our axioms.   
 
Clearly, our requirements are tied to our goal of obtaining a
definition that can be used in practical situations. 
An open question is whether the properties put forth here as desirable
are the only ones---or at least the \emph{important} ones, in some sense of
\emph{important}. We believe that our properties satisfy
intuitions that make the concept of FD usable in practical
situations. However,  additional properties should be explored to get
a better understanding of the \emph{desired} or \emph{expected} behavior of FDs
in the presence of \nullname{} markers; perhaps a set of alternative
properties, each giving raise to a concept of FD, can be developed for
different scenarios.
Agreeing on some set or sets of basic requirements would
provide researchers with an explicit milestone by which to judge
different formalizations. 

A narrower question is whether the definitions proposed here are the
only ones that can satisfy all the requirements put forth, or whether
alternatives exist. While the authors have considered many alternative definitions, and found them missing some requirement,
it is not known whether other definitions could exist that would still
satisfy all properties, and if so, what would the relationships
between different definitions be. 

\section{Funding} 

This work was supported by the Natural Sciences and Engineering Research Council of Canada [26143].

\section{Acknowledgements} 

We thank Carles~Farr\'e and Andre Vellino for their helpful comments.

\bibliographystyle{compj}
\bibliography{nullfd}

\end{document}